\documentclass[11pt,letterpaper]{article}
\usepackage[T1]{fontenc}
\usepackage[utf8]{inputenc}
\usepackage{mathtools,amssymb,amsthm}
\usepackage[margin=1in]{geometry}
\usepackage{booktabs,array}
\usepackage{graphicx}
\usepackage{placeins}
\usepackage{algorithm,algpseudocode}
\usepackage[numbers,sort&compress]{natbib}
\usepackage[hyphens]{url}
\usepackage[hidelinks]{hyperref}
\usepackage[nameinlink,noabbrev]{cleveref}

\makeatletter
\providecommand{\theHALG@line}{}
\renewcommand{\theHALG@line}{\thealgorithm.\arabic{ALG@line}}
\makeatother

\hypersetup{pdftitle={Deterministic Linear-Time Modular Subset Sum},
  pdfauthor={Phuoc Dinh Le and Kha Le},
  pdfsubject={Deterministic modular subset sum with compact multiplicities},
  pdfkeywords={modular subset sum, cyclic intervals, deterministic algorithms}}

\newtheorem{theorem}{Theorem}[section]
\newtheorem{lemma}[theorem]{Lemma}
\newtheorem{corollary}[theorem]{Corollary}

\crefname{theorem}{theorem}{theorems}
\Crefname{theorem}{Theorem}{Theorems}
\crefname{lemma}{lemma}{lemmas}
\Crefname{lemma}{Lemma}{Lemmas}
\crefname{corollary}{corollary}{corollaries}
\Crefname{corollary}{Corollary}{Corollaries}
\crefname{proposition}{proposition}{propositions}
\Crefname{proposition}{Proposition}{Propositions}
\theoremstyle{definition}

\theoremstyle{remark}

\newcommand{\Zm}{\mathbb Z_m}
\newcommand{\Oh}{\mathcal O}
\newcommand{\SortCost}{\operatorname{Sort}}

\newcommand{\Reach}{\mathsf{reached}}
\newcommand{\Par}{\mathsf{parent}}
\newcommand{\Val}{\mathsf{value}}
\newcommand{\Change}{\mathsf{change\_direction}}
\newcommand{\Extend}{\mathsf{extend\_and\_merge}}

\newcommand{\EMPTY}{\mathsf{EMPTY}}
\newcommand{\FULL}{\mathsf{FULL}}
\newcommand{\PARTIAL}{\mathsf{PARTIAL}}
\newcommand{\START}{\mathsf{START}}
\newcommand{\END}{\mathsf{END}}

\title{Deterministic Linear-Time Modular Subset Sum}
\author{%
  Phuoc Dinh Le\\
  \small Georgia Institute of Technology\\
  \small\href{mailto:lephuocdinh99@gmail.com}{\texttt{lephuocdinh99@gmail.com}}
  \and
  Kha Le\\
  \small Texas A\&M University\\
  \small\href{mailto:lephuocanhkha2003@gmail.com}{\texttt{lephuocanhkha2003@gmail.com}}}
\date{27 October 2026}

\begin{document}
\maketitle

\begin{abstract}
We give a deterministic $\Oh(m)$-time algorithm for exact modular
subset sum over every modulus $m$. The input gives distinct residues
and their multiplicities. The algorithm reports all reachable residues
and answers one target query. It returns a witness when the target is reachable.
The algorithm uses $\Oh(m)$ auxiliary words on an arithmetic word-RAM.
This improves Pot\k{e}pa's deterministic
$\Oh(m\log m\,\alpha(m))$-time bound under the same input convention.
The running time matches the cost of returning a length-$m$ Boolean array.

We represent reachable residues as runs along cycles of repeated addition.
Newly reached residues pay for scans of partial runs. Packed cycle flags
make rebuilding take linear total time at the original modulus. A theorem on
subset sums of distinct units limits the number of adaptive boundary
batches to $\Oh(m^{3/4})$. Radix sorting their $\Oh(m)$ total keys also
takes $\Oh(m)$ time.
\end{abstract}

\section{Introduction and main result}\label{sec:introduction}

Exact modular subset sum asks which residues in $\Zm=\mathbb Z/m\mathbb Z$
are \emph{reachable} as sums of input copies. Each copy may be used at
most once. We allow the empty subset.
We use the compact input convention of Pot\k{e}pa~\cite{Potepa2021}:
at most $m$ records $(x,c_x)$, where the residues $x$ are distinct and
$c_x$ is the nonnegative number of available copies of $x$.
For this input convention, Pot\k{e}pa~\cite[Theorem~2]{Potepa2021}
gives a deterministic $\Oh(m\log m\,\alpha(m))$-time algorithm for all
reachable residues using $\Oh(m)$ space, where $\alpha$ is the inverse
Ackermann function. Our result removes the superlinear factor for every
modulus. Since we return a Boolean array of length $m$, $\Oh(m)$ time
matches the cost of reporting the full reachable set.

\begin{theorem}\label[theorem]{thm:main}
For every $m\ge1$, exact modular subset sum on compact input can be solved
deterministically in $\Oh(m)$ time and $\Oh(m)$ auxiliary words under the
arithmetic model specified below. The algorithm returns the
Boolean array of all reachable residues. It also answers one query for a
target $t$ within the same bounds. It reports that $t$ is
unreachable or returns witness counts $b_x$ satisfying
\[
0\le b_x\le c_x,\qquad
\sum_x b_xx\equiv t\pmod m,\qquad
\sum_x b_x\le m-1.
\]
Alternatively, the witness can identify each selected copy by its input
record and its copy number within that record.
\end{theorem}

\begin{corollary}\label[corollary]{cor:explicit}
For an explicit list of $n$ reduced residues that can be scanned twice, the time is
$\Oh(n+m)$ and the auxiliary space is $\Oh(m)$ words, excluding the
original input storage. This includes one requested witness using distinct
original list indices.
\end{corollary}

\paragraph{Prior work and attribution.}
Axiotis et al.~\cite{AxiotisEtAl2019} give a randomized near-linear
algorithm. They accelerate Bellman's recurrence with linear sketches.
Its $\widetilde O(m)$ bound hides logarithmic factors: the explicit bound
is $\Oh(m\log^7 m)$~\cite{AxiotisEtAl2021}.
Cardinal and Iacono give an expected $\Oh(m\log m)$
algorithm. We use the expected-time statement of their revised full
version~\cite{CardinalIacono2021}. Axiotis et al.~\cite{AxiotisEtAl2021}
give further algorithms through dynamic strings, including a deterministic
bound that depends on the number of reachable residues.
Bringmann and Nakos~\cite{BringmannNakos2021} study the more general
sumset $A_1+\cdots+A_n\subseteq\mathbb Z_m$. They give an almost-linear
deterministic algorithm and a Las Vegas algorithm with near-linear expected
time. Both bounds depend on the input-plus-output size. Modular subset
sum is the case $A_i=\{0,a_i\}$. We obtain a linear deterministic bound
for this special case on compact input.

Arithmetic progressions also appear in the constructive
Erd\H{o}s--Ginzburg--Ziv algorithms of Leung and Jo. Jo's prime-modulus
subroutine finds a witness for any target from $p-1$ nonzero residues
in $\Oh(p)$ time~\cite{Leung2025,Jo2026}.
These results motivate arithmetic-progression representations. Our
algorithm uses cycle runs to compute exact reachability for unrestricted
input and arbitrary moduli.

Grouping residues by divisors also appears in the work of Koiliaris and
Xu~\cite{KoiliarisXu2019}. Recording a parent when a residue is first reached
is used to construct witnesses in earlier modular
algorithms~\cite{CardinalIacono2021}. The theorem on subset sums of distinct
invertible residues used here is due to DeVos, Goddyn, Mohar, and
\v{S}\'amal~\cite{DeVosEtAl2007}, following Vu~\cite{Vu2005}.

Unlike Pot\k{e}pa's shift-tree approach, we organize updates by gcd class
and store cycle states in packed bitsets. We use a theorem from additive combinatorics
to bound sorting. The updates use arrays, interval lists, and radix sort.
Across our benchmark inputs,
median runtimes give speedups of about $4.5$--$22\times$ over the hashed
shift tree and $18$--$48\times$ over the deterministic shift tree
(\cref{sec:validation}). From our review of the literature, this is the
first deterministic $\Oh(m)$-time algorithm for unrestricted exact modular
subset sum on compact input over arbitrary moduli.

\paragraph{Model and input conventions.}\label{sec:model}
We work with residues in $[0,m)$ on a word-RAM whose words hold
$\Omega(\log(m+1))$ bits and that supports constant-time comparisons, bit
operations, array access, and arithmetic, the last including integer
quotient, remainder, and modular multiplication. Each compact record and
intermediate product fits in a constant number of words. To enumerate
reachable residues we pack flags into blocks of
$w=\max\{1,\lfloor(\log_2 m)/2\rfloor\}$ bits. At half of $\log_2 m$, a
table indexed by one block has only $2^w=\Oh(\sqrt m)$ entries, so we build
it once in linear time. It returns the lowest set position of a nonzero
block. Clearing that bit takes one Boolean operation,
so enumerating $q$ set bits among $a$ packed bits costs $\Oh(1+a/w+q)$. We state all packed
costs in units of these $w$-bit blocks. Machine words may be wider, so this
can only overcount, and $w=\Theta(\log m)$ is all the bounds below require.
Finally, we cap each nonzero record at
$\widehat c_x=\min\{c_x,m/\gcd(x,m)-1\}$ and discard the zero records
(\cref{lem:cap}).

\paragraph{Proof overview.}
We store reachable runs along addition cycles. Changing the input direction
recovers run boundaries from the difference between the set and its shift.
An update extends each run by one. Its growth bounds the number of runs
left for the next scan, so all scans cost $\Oh(m)$
(\cref{sec:prime,sec:representation}). Grouping gcd classes and packing
constant-cycle flags makes cycle changes cost $\Oh(m)$ as well
(\cref{sec:filtration}).

There are $\Oh(m)$ boundary keys in total. A theorem on subset sums of
distinct units bounds the number of adaptive sorting batches by
$\Oh(m^{3/4})$, allowing radix sort to process them in linear total time
(\cref{sec:sorting}). Gcd tables, batch inversion, and first-discovery
parents complete the preprocessing and witness bounds
(\cref{sec:preprocessing,sec:witness}).

\section{The core algorithm: one cycle}\label{sec:prime}

Let $m=p$ be prime. Starting from the empty subset's sum $S=\{0\}$,
we use the exact Bellman update
\begin{equation}\label{eq:bellman}
 S\longleftarrow S\cup(S+x),\qquad S+x=\{s+x\bmod p:s\in S\}.
\end{equation}
For a nonzero \emph{direction} $d$, store $S$ as maximal occupied runs
in the order $0,d,\ldots,(p-1)d$. Position $t$ represents $td\bmod p$,
and residue $z$ has position $zd^{-1}\bmod p$.
Runs are stored as ordered half-open intervals in $[0,p)$, splitting a
wrapping run at zero. The full cycle uses a flag.
We also keep membership bits $\Reach[z]$ and the cyclic run count
$r_d(S)$, with $r_d(\mathbb Z_p)=0$.
One update changes the representation to direction $x$, then extends
its runs by one position:
\begin{equation}\label{eq:twooperations}
 (S,\mathcal I_d)
 \xrightarrow{\ \Change(d,x)\ }(S,\mathcal I_x)
 \xrightarrow{\ \Extend(x)\ }(S\cup(S+x),\mathcal I'_x).
\end{equation}

\begin{figure}[H]
\centering
\includegraphics[width=\textwidth]{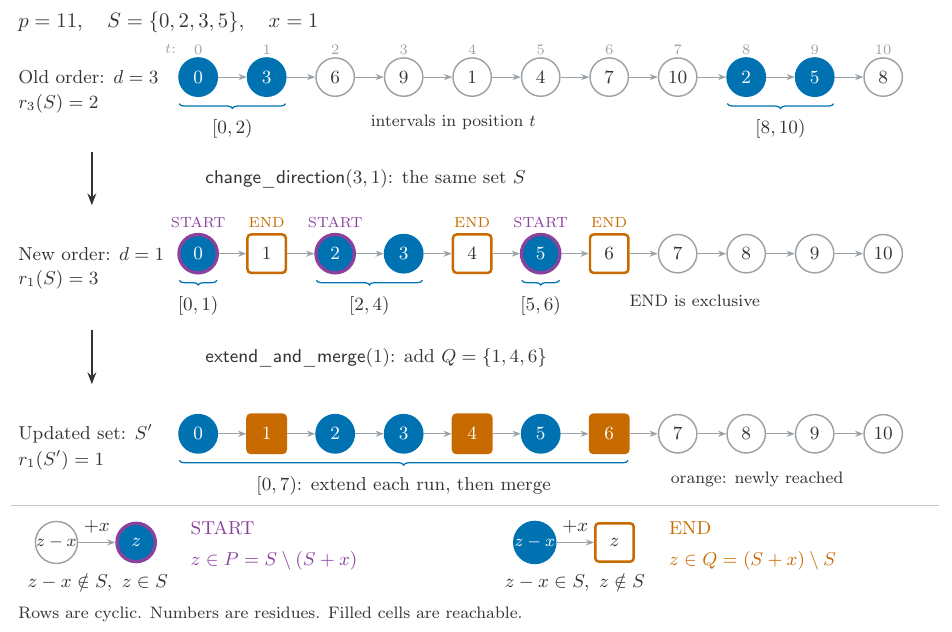}
\caption{One update modulo $11$. The subset sums of $(2,3)$ are
$S=\{0,2,3,5\}$. Changing direction $3\to1$ gives three runs of the
same set. Extension adds $Q=\{1,4,6\}$ and merges them into one.
Cell labels are residues. Interval endpoints are positions.}
\label{fig:prime}
\end{figure}

The predecessor of $z$ in direction $x$ is $z-x$, so
\begin{align}
 z\text{ starts a run}&\iff z\in S,\ z-x\notin S,\nonumber\\
 z\text{ is an exclusive end}&\iff z\notin S,\ z-x\in S.
 \label{eq:predecessor}
\end{align}
Thus the new run starts and exclusive ends are exactly
\begin{equation}\label{eq:differences}
 P=\underbrace{S\setminus(S+x)}_{\START},\qquad
 Q=\underbrace{(S+x)\setminus S}_{\END},\qquad
 \delta=|Q|=|P|.
\end{equation}
Here $\delta$ is the update's growth, and $|P|=|Q|$ follows from
$|S+x|=|S|$. Coordinate inverses are precomputed in $\Oh(p)$ time
(\cref{lem:preprocessing}). Write $\SortCost(k)$ for the cost of sorting
one batch of $k$ tagged boundaries, with $\SortCost(0)=0$.

\begin{lemma}[Direction change]\label[lemma]{lem:primechange}
For $\varnothing\ne S\ne\mathbb Z_p$, changing direction $d\to x$
preserves $S$ and costs $\Oh(1+r_d(S)+\delta)+\SortCost(2\delta)$.
If $d=x$, no work is needed.
\end{lemma}
\begin{proof}
In direction $d$, translation by $x$ shifts every interval by
$k=xd^{-1}\bmod p$. Split wrapping intervals, rotate their list to
restore order, and merge touching pieces. A two-pointer difference sweep of the old and shifted
lists enumerates $P,Q$ in $\Oh(1+r_d(S)+\delta)$ time.
Tag each residue in $P$ as $\START$ and each in $Q$ as $\END$.
Convert $z$ to position $zx^{-1}\bmod p$ and sort these $2\delta$ events.
By \cref{eq:predecessor}, tags alternate cyclically. Pairing each start
with the next end reconstructs precisely the old occupied runs in
$\Oh(\delta)$ time. A wrapping pair $(s,e)$ gives $[s,p)$ and $[0,e)$,
omitting the latter when $e=0$. Membership is unchanged.
\end{proof}

\begin{lemma}[Growth pays for the next scan]\label[lemma]{lem:primegrowth}
Extending each run by one and merging touching runs gives exactly
$S'=S\cup(S+x)$ in $\Oh(1+r_x(S)+\delta)$ time, with
\begin{equation}\label{eq:primerunbound}
 r_x(S')\le\delta.
\end{equation}
\end{lemma}
\begin{proof}
The old exclusive ends are exactly $Q$. Extension adds these distinct
positions, so it performs the Bellman update. Extending and merging
cannot split a run, giving $r_x(S')\le r_x(S)=|Q|=\delta$.
A completed cycle has run count zero. The ordered extension and merge
use one sweep. Before marking each $z\in Q$ reachable, record its old
reachable parent $z-x$ and the copy used.
\end{proof}

Let $S_i$ be the set after update $i$, with growth $\delta_i$.
Every nonzero direction generates $\mathbb Z_p$, so a nonempty proper
set must grow. At most $p-1$ updates execute before full coverage.
With $S_0=\{0\}$ and $x_0=1$, the preceding lemmas give
\begin{equation}\label{eq:primeamortization}
 r_{x_i}(S_i)\le\delta_i,\qquad
 \sum_i\delta_i\le p-1,\qquad
 \sum_i(1+r_{x_{i-1}}(S_{i-1})+\delta_i)=\Oh(p).
\end{equation}
The first scan starts from one run. Process equal values consecutively:
after their first copy, the direction is already $x$, so further copies
use only $\Extend(x)$. Thus all nonsorting work is linear, and at most
one boundary batch is sorted per distinct record.
\Cref{sec:sorting} proves that all sorting also costs $\Oh(p)$.

\section{Arbitrary moduli: several cycles}\label{sec:representation}

For a nonzero direction $d$, put $g=\gcd(d,m)$, $L=m/g$,
$u=d/g$, and $v=u^{-1}\bmod L$. Its cycles are the residue classes
modulo $g$. Cycle $c$ has coordinates
\begin{equation}\label{eq:rowmap}
 C_c=c+g\mathbb Z_L,\qquad
 \phi_d(c,t)=c+g(ut\bmod L),\quad 0\le c<g,\quad 0\le t<L.
\end{equation}
Adding $d$ advances one position within a cycle. Every cycle is
\emph{empty}, \emph{full}, or \emph{partial}, according to whether it holds
no residue of $S$, all $L$ of them, or a proper mixture. We mark the state
of all $g$ cycles in packed masks $\EMPTY,\FULL,\PARTIAL$ of length $g$.
A full or empty cycle needs nothing more, as it has no boundaries. A
partial cycle also stores its occupied positions as runs, at most two
ordinary intervals per cyclic run, and $R_d(S)$ denotes the total run count.
Each partial cycle holds at least one run, so at most $R_d(S)$ are partial.
We keep them in an explicit list and reach or release their run records only
through it. Two length-$m$ pointer arrays,
allocated once, index the old and new partial records across a direction
change. We also keep packed membership $\Reach$ and $|S|$.

\begin{figure}[H]
\centering
\includegraphics[width=\textwidth]{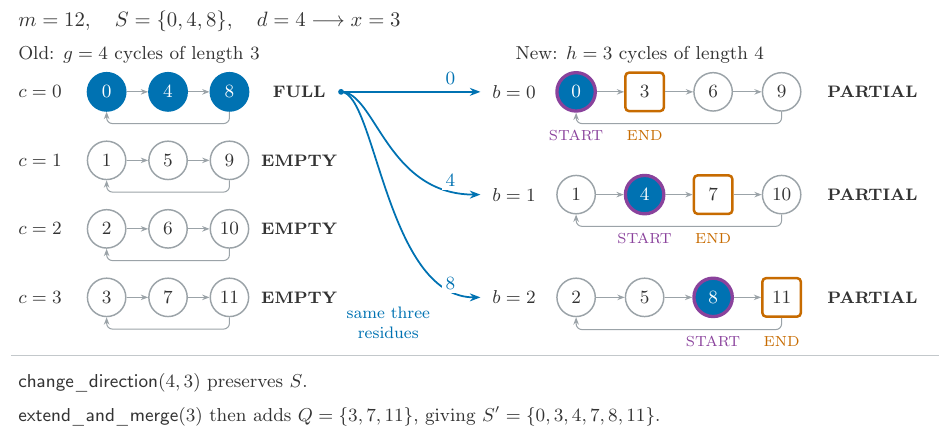}
\caption{Changing direction $4\to3$ modulo $12$ turns the full cycle
$\{0,4,8\}$ into three partial cycles. Both panels show the same set.
Extending in direction $3$ then adds $3,7,11$.}
\label{fig:composite}
\end{figure}

\begin{lemma}[Exact cycle update]\label[lemma]{prop:update}
Let $h=\gcd(x,m)$ and $\delta=|(S+x)\setminus S|$. Changing direction
$d\to x$ and applying one copy of $x$ takes
\[
 \Oh\!\left(1+R_d(S)+\delta+
 \mathbf1_{\{g\ne h\}}(g+h)/w\right)+\SortCost(k),
 \qquad k=\begin{cases}2\delta,&d\ne x,\\0,&d=x.\end{cases}
\]
The update is exact and records valid first-discovery parents.
\end{lemma}
\begin{proof}
The plan mirrors the one-cycle case: find the boundaries that change,
sort them, rebuild the old runs, then extend and record parents.

If $d=x$ there is no direction change. Extend each partial run by one as
in \cref{lem:primegrowth}, leaving full and empty cycles untouched. Assume
$d\ne x$ from now on, and split on whether the shift preserves the gcd,
that is, whether $g=\gcd(d,m)$ equals $h=\gcd(x,m)$.

Suppose first that $g=h$. Then $g\mid x$, so $c'=(c+x)\bmod g=c$ and both
directions induce the same cycle partition. Full and empty cycles, and all
flags, stay as they are. Only the partial cycles change, and inside each we
run the one-cycle direction change and extension of
\cref{lem:primechange,lem:primegrowth}, batching the at most $2\delta$
boundaries into one sort. This costs
$\Oh(1+R_d(S)+\delta)+\SortCost(2\delta)$.

Now suppose $g\ne h$. Adding $x$ sends cycle $c$ to a cycle $c'$ and
advances its positions by $k_c$,
\begin{equation}\label{eq:carry}
 c'=(c+x)\bmod g,\qquad
 k_c=\frac{c+x-c'}g\,v\bmod L,
\end{equation}
because $\phi_d(c,t)+x\equiv\phi_d(c',t+k_c)\pmod m$. As in
\cref{lem:primechange} we shift each cycle's interval list by $k_c$, split
a wrapping interval at zero, and sweep it against the list at its
destination. A full cycle contributes $[0,L)$ and an empty cycle the empty
list. This recovers the boundaries $P=S\setminus(S+x)$ and
$Q=(S+x)\setminus S$.

A cycle changes only if it or the cycle mapped onto it is partial, or a
full cycle lands on an empty one or the reverse. Let $\rho$ rotate a
$g$-bit mask forward by $x\bmod g$. A non-partial cycle changes exactly where $\FULL$
and its rotation differ, so the candidates are
\begin{equation}\label{eq:candidates}
 A=\PARTIAL\lor\rho\PARTIAL\lor(\FULL\oplus\rho\FULL).
\end{equation}
Each set bit of $A$ names a destination cycle $c'$, with source
$c=(c'-x)\bmod g$. At most $2R_d(S)$ of them touch a partial cycle. Every other candidate is a
full cycle moving onto an empty one, contributing a whole cycle to $P$ or
$Q$, so there are at most $2\delta$ of those. The rotation is a packed
shift, OR, and mask, costing $\Oh(1+g/w)$, and each partial list is swept
at most twice. Enumerating the $2\delta$ events costs
$\Oh(1+R_d(S)+\delta+g/w)$.

Give each boundary residue $z$ its new cycle, position, and sort key:
\begin{equation}\label{eq:newcoordinates}
 K=m/h,\quad v_x=(x/h)^{-1}\bmod K,\qquad
 b=z\bmod h,\quad t=\frac{z-b}{h}\,v_x\bmod K,\quad \mathrm{key}=bK+t.
\end{equation}
The $2\delta$ keys are distinct and lie in $[0,m)$, and sorting them orders
the events by cycle and then position. In each cycle that bore events, the
tags alternate cyclically, so pairing each start with the next end rebuilds
its runs. These event-bearing cycles are the new partial cycles, holding
$\delta$ runs in total.

Build the new $\PARTIAL$ mask by clearing $h$ bits and setting the bit of
each new partial cycle. Residue $b$ is position zero of cycle $b$, so with
$B=\Reach[0:h)$ the other two flags are
\begin{equation}\label{eq:rebuildflags}
 \FULL=B\land\neg\PARTIAL,\qquad
 \EMPTY=(\neg B)\land\neg\PARTIAL,
\end{equation}
with complements taken over $h$ bits. Building the new partial records
costs $\Oh(\delta)$, releasing the old ones $\Oh(R_d(S))$, and swapping the
masks $\Oh(1+(g+h)/w)$.

Finally, we extend the $\delta$ rebuilt runs by one and merge, as in
\cref{lem:primegrowth}. For each old exclusive end $z\in Q$ record
\begin{equation}\label{eq:parents}
 \Par[z]=(z-x)\bmod m\in S,\qquad \Val[z]=x,
\end{equation}
writing the parent before marking $z$ reachable, even when the update
completes a cycle. By the one-cycle argument the result is exactly
$S\cup(S+x)$. When $\delta=0$ there are no events, so $\PARTIAL$ is empty,
the flags reduce to $\FULL=B$ and $\EMPTY=\neg B$, and no partial run
remains.
\end{proof}

\begin{lemma}[Growth pays for the next scan]\label[lemma]{lem:growth}
In direction $x$, an update leaves at most $\delta$ partial runs, and each
further copy of $x$ costs $\Oh(1+R_x(S)+\delta)$ with no sorting.
\end{lemma}
\begin{proof}
Summing \cref{lem:primegrowth} over the partial cycles leaves at most
$\delta$ runs, since their growths sum to $\delta$. A repeated $x$ keeps
the direction fixed, so with no boundary to sort it only extends the
$R_x(S)$ runs by one and marks the $\delta$ new residues, leaving full and
empty cycles untouched.
\end{proof}

\section{Grouping gcd classes}\label{sec:filtration}

Reordering the input preserves its final subset sums. Let
$G=\{\gcd(x,m):x\ne0,\ c_x>0\}$ be the present gcd classes.
We process each class consecutively in a \emph{gcd phase}, in any order,
and process equal values consecutively. All updates use the original modulus.

\begin{algorithm}[H]
\caption{Exact modular subset sum on compact input}\label{alg:main}
\begin{algorithmic}[1]
\Require $m\ge1$, distinct reduced records $(x,c_x)$
\If{$m=1$} \State \Return the membership array $[1]$ and the empty witness \EndIf
\State Cap counts, form gcd classes, and compute their reduced inverses
\State Initialize $S=\{0\}$ in direction $d\gets1$, packed flags, and sorting workspace
\For{each present gcd class $h$}
 \For{each record $(x,\widehat c_x)$ in this class}
  \State $r\gets\widehat c_x$
  \While{$r>0$ and $|S|<m$ and $(\gcd(d,m)\ne h\text{ or }R_d(S)>0)$}
   \State $\Change(d,x)$ \Comment{leaves $d=x$}
   \State $\Extend(x)$ \Comment{record parents and mark new membership bits}
   \State $r\gets r-1$
  \EndWhile
 \EndFor
\EndFor
\State \Return the membership array and first-discovery arrays $\Par,\Val$
\end{algorithmic}
\end{algorithm}

\begin{lemma}[Packed cycle rebuilding]\label[lemma]{lem:transitions}
Handling full and empty cycles at gcd-class changes takes $\Oh(m)$
time in total, apart from work charged to partial runs and discoveries.
\end{lemma}
\begin{proof}
By \cref{prop:update}, a change from $g$ to $h$ cycles
costs $\Oh(1+(g+h)/w)$ beyond that work.
Each present gcd occurs at most once as the new gcd and once as the old.
The initial direction contributes $1$. Writing $\tau(m)$ and $\sigma(m)$
for the number and sum of divisors, the total is
\[
 \Oh\!\left(\tau(m)+\frac{\sigma(m)}w\right)=\Oh(m).
\]
Indeed, $\tau(m)\le m$ and the elementary harmonic bound gives
\[
 \sigma(m)=m\sum_{a\mid m}\frac1a
 \le m\sum_{a=1}^m\frac1a\le m(1+\ln m).
\]
Since $w=\Theta(\log(m+1))$, this suffices for linear time.
\end{proof}

\begin{lemma}[Total cost of updates and scans]\label[lemma]{lem:linearruns}
Skipping records as in \cref{alg:main} preserves the reachable set.
At most $2(m-1)$ keys are sorted in total. Loop checks, interval scans,
boundary enumeration, merging, and parent recording take $\Oh(m)$ total time.
\end{lemma}
\begin{proof}
The set only grows, so $\sum_i\delta_i\le m-1$.
If no partial cycle remains for a direction with gcd $h$, then $S$ is a
union of cosets of $h\mathbb Z_m$. Every remaining value $x$ in that gcd
class lies in $h\mathbb Z_m$, so $S+x=S$.
A zero-growth update leaves $R_x(S')=0$ by \cref{lem:growth}, so at most
one occurs per class. Hence
\[
 \#\text{executed updates}\le m-1+|G|=\Oh(m).
\]
The stored gcd and partial-cycle list make each guard constant time.
Each record adds at most one failed check, and each direction change
contributes $2\delta$ keys.
Initially there is one run. Each later update scans at most the preceding
update's growth by \cref{lem:growth}, giving
\[
 \sum_{\text{updates}}\text{runs before the update}
 \le1+\sum_i\delta_i\le m.
\]
The scans cover merging, flag maintenance, and partial-record replacement.
Boundary enumeration, membership writes, and parent recording cost
$\Oh(\sum_i\delta_i)$.
\end{proof}

Let $T$ be the total running time and $k_i$ the number of keys sorted in
update $i$, with $k_i=0$ when no sorting is needed. Combining the bounds gives
\begin{equation}\label{eq:sortingbudget}
 \boxed{\begin{gathered}
 V=\sum_i k_i\le2\sum_i\delta_i\le2(m-1),\\
 T=\Oh(m)+\sum_i\SortCost(k_i).
 \end{gathered}}
\end{equation}
It remains to sort these adaptively generated batches.

\section{Linear total sorting time}\label{sec:sorting}

Repeated copies need no direction change, so each distinct record
produces at most one boundary batch. We bound the number of processed
records using a completeness theorem. Write $\Sigma(A)$ for all subset
sums of a set $A$, including the empty sum.

\begin{theorem}[DeVos et al.~\cite{DeVosEtAl2007}]\label[theorem]{thm:units}
If $L\ge2$ and $A\subseteq\mathbb Z_L$ consists of distinct units with
$|A|\ge8\sqrt L$, then $\Sigma(A)=\mathbb Z_L$.
\end{theorem}
This is the unit-set consequence following their Theorem~1.5.

\begin{lemma}[Processed records per class]\label[lemma]{lem:classrecords}
A gcd phase $g$ processes at most $\lceil8\sqrt{m/g}\rceil$ distinct records.
\end{lemma}
\begin{proof}
Set $L=m/g\ge2$. Distinct records reduce to distinct units $x/g$ modulo
$L$. Let $S_0$ be the phase-entry set and $A$ the reduced units whose
first copies have been processed. Those copies are disjoint from the
copies forming $S_0$, so exact updates give $S\supseteq S_0+g\Sigma(A)$.
Once $|A|\ge8\sqrt L$, \cref{thm:units} fills every cycle meeting $S_0$.
Updates are multiples of $g$, so initially empty cycles remain empty.
No partial cycle remains, and the rest of the phase is skipped.
\end{proof}

\begin{lemma}[A divisor estimate]\label[lemma]{lem:weighted}
For every $m\ge1$,
\[
 \sum_{g\mid m}g^{-1/2}\le C_0m^{1/4},\qquad
 C_0=\prod_{q\in\{2,3,5\}}(1-q^{-1/2})^{-1}.
\]
\end{lemma}
\begin{proof}
The sum factors as $\prod_{q^e\parallel m}\sum_{j=0}^e q^{-j/2}$.
For each prime $q\ge7$ and $e\ge1$,
\[
 \sum_{j=0}^e q^{-j/2}\le(1-q^{-1/2})^{-1}\le q^{1/4}\le q^{e/4}.
\]
The middle inequality is equivalent to $\sqrt q-1\ge q^{1/4}$,
which holds since $q^{1/4}\ge(1+\sqrt5)/2$ for $q\ge7$.
For $q=2,3,5$, use the infinite geometric series. Multiplying proves the bound.
\end{proof}

\begin{corollary}[Number of batches]\label[corollary]{cor:batches}
The number $B$ of nonempty sorting batches is $\Oh(m^{3/4})$.
\end{corollary}
\begin{proof}
Each gcd class occurs once, and only first copies can require sorting.
By \cref{lem:classrecords,lem:weighted},
\[
 B\le\sum_{\substack{g\mid m\\g<m}}\lceil8\sqrt{m/g}\rceil
 \le9\sqrt m\sum_{g\mid m}g^{-1/2}=\Oh(m^{3/4}).
\]
\end{proof}

Define $\mathsf{sort\_events}$ as LSD radix sort on keys in $[0,m)$,
retaining each residue and tag. With base
$D=2^{\lceil(\log_2m)/8\rceil}=\Theta(m^{1/8})$, at most eight passes
sort a nonempty batch of $k$ keys in $\Oh(k+D)$ time, including counter resets.

\begin{lemma}[Sorting all batches]\label[lemma]{lem:sorting}
All sorting takes $\Oh(m)$ deterministic time and $\Oh(m)$ auxiliary words.
\end{lemma}
\begin{proof}
The total key volume is $V\le2(m-1)$ by \cref{eq:sortingbudget}.
Together with \cref{cor:batches}, this gives
\[
 \sum_{i:k_i>0}\Oh(k_i+D)=\Oh(V+BD)=\Oh(m+m^{7/8})=\Oh(m).
\]
The batches share one counter array and one scratch array, initialized
once in $\Oh(m)$ time and using $\Oh(m)$ words.
\end{proof}

\section{Preprocessing}\label{sec:preprocessing}

\begin{lemma}[Capping copies]\label[lemma]{lem:cap}
For $x\ne0$, capping $c_x$ at $m/\gcd(x,m)-1$ preserves all subset sums.
Zero records may be discarded.
\end{lemma}
\begin{proof}
The order of $x$ is $L=m/\gcd(x,m)$. Replace any chosen number of
copies $j$ by $j\bmod L\le\min\{j,L-1\}$ without changing its sum.
\end{proof}

\begin{lemma}[Gcd classes and inverses]\label[lemma]{lem:preprocessing}
All residue gcds and reduced input inverses used above can be computed
in $\Oh(m)$ time and words.
No factorization of $m$ is assumed.
\end{lemma}
\begin{proof}
A linear sieve generates each composite $pz<m$ once with its smallest
prime factor, using primes $p\le\operatorname{spf}(z)$. It builds the
smallest-prime-factor table in $\Oh(m)$ time. Set $G(1)=1$.
For $2\le z<m$, let $p=\operatorname{spf}(z)$ and $a=G(z/p)$. Then
\begin{equation}\label{eq:gcdtable}
 G(z)=\begin{cases}pa,&p\mid(m/a),\\a,&\text{otherwise}.\end{cases}
\end{equation}
Multiplication by $p$ raises the gcd's $p$-exponent exactly when $m/a$
still contains $p$. Thus $G(z)=\gcd(z,m)$ by induction. A scan groups
records by $G(x)$ and caps their counts. We never sum multiplicities.

For each present gcd class $g$, put $L=m/g$ and $u_j=x_j/g$ for its
$k$ records. Form prefix products $\Pi_0=1$,
$\Pi_j=\Pi_{j-1}u_j\bmod L$, invert $\Pi_k$ by extended Euclid,
and recover all individual inverses backwards:
\[
 T\gets\Pi_k^{-1}\bmod L,\qquad
 \text{for }j=k,\ldots,1:\quad
 u_j^{-1}\gets\Pi_{j-1}T\bmod L,\quad T\gets Tu_j\bmod L.
\]
The invariant before iteration $j$ is $T=\Pi_j^{-1}$.
This costs $\Oh(k+\log L)$ per class. There are at most
$\tau(m)\le2\sqrt m$ classes and at most $m-1$ nonzero records, so the
total is $\Oh(m+\sqrt m\log(m+1))=\Oh(m)$. All tables use $\Oh(m)$ words.
\end{proof}

\section{Witnesses and total complexity}\label{sec:witness}

\begin{proof}[Proof of \cref{thm:main}]
Capping preserves every record's contribution by \cref{lem:cap}.
The updates are exact by \cref{prop:update}, and skipped copies are
redundant by \cref{lem:linearruns}. Starting from $\{0\}$ therefore gives
exactly the final reachable set. Expand the packed membership bits to a
length-$m$ Boolean array in $\Oh(m)$ time.

If $\Reach[t]=0$, report the target unreachable. Otherwise follow
\cref{eq:parents} from $t$ to zero. Each parent was reached before the
update that discovered its child, so discovery times strictly decrease.
The path has at most $m-1$ edges and uses each update, hence each input
copy, at most once. Counting edges labelled $x$ gives $b_x\le c_x$,
and the parent equations telescope to $\sum_x b_xx\equiv t\pmod m$.
Return these counts or the first $b_x$ copy labels in each record.
Target zero uses the empty path. Reconstruction costs $\Oh(m)$ time.

Preprocessing, packed cycle changes, scans, and sorting each take
$\Oh(m)$ time by \cref{lem:preprocessing,lem:transitions,lem:linearruns,lem:sorting}.
Counts, membership, parents, inverse tables, buckets, partial records,
and shared sorting workspace use $\Oh(m)$ words in total, including
consecutive old and new representations.
\end{proof}

\begin{proof}[Proof of \cref{cor:explicit}]
Build a capped histogram in $\Oh(n+m)$ time and apply \cref{thm:main}.
Rescan the list and take the first $b_x$ indices of each residue $x$.
This produces distinct original indices within the claimed bounds,
excluding original input storage.
\end{proof}

\section{Implementation and experiments}\label{sec:validation}

We compare our fixed-modulus C++17 implementation with packed cycle flags
against basic 64-bit bitset DP and the
hashed and deterministic shift-tree algorithms of
Pot\k{e}pa~\cite{Potepa2021}.
All four compute the full reachable set and witness parents. The hashed
shift tree is Monte Carlo. The other three algorithms are deterministic.
With zero-growth skipping, bitset DP executes $\Oh(m)$ updates, each
scanning $\lceil m/w\rceil$ words. Its worst-case time is
$\Oh(m+m^2/w)$, with $w=64$.
Bitset DP's reachability state uses $\Oh(m)$ bits, whereas our cycle
representation uses $\Oh(m)$ words. All benchmark solvers also
store $\Oh(m)$ words for witnesses.

\paragraph{Validation.}
We compare each implementation with an independent Boolean DP on 3,883
inputs: all 3,279 multiplicity vectors with $1\le m\le7$ and counts in
$\{0,1,2\}$,
404 targeted cases at moduli up to $1{,}025$, and 200 seeded random cases.
The targeted cases include empty inputs, very large counts, prime powers,
and changes of gcd class. For each run, we check the full reachable set.
We replay every growth-producing update against the Bellman recurrence
in the solver's processing order. We also query every target. Each returned witness
must use available copies and sum to its target. Unreachable targets must
return no witness.
The packed-cycle implementation also passes direct checks of cycle flags,
reachability preservation under direction changes, and cyclic bit rotations
across word boundaries.

\paragraph{Benchmarks.}
We use five size anchors:
\[
A\in\{1,\,2.5,\,5,\,7.5,\,10\}\times10^5.
\]
For each, we choose the largest prime at most $A$ and the composite
$30\lfloor A/30\rfloor$. At each modulus $m$, we generate three inputs:
\begin{itemize}
\item \emph{Random support:} $\lfloor m/4\rfloor$ distinct nonzero residues,
      sampled uniformly, with one copy of each.
\item \emph{Repeated generator:} the single compact record $(1,m-1)$.
\item \emph{Powers of three:} one copy of each $3^j<m$, starting with $3^0=1$.
\end{itemize}
The algorithms receive the same 30 fixed inputs. Each solver processes records
in the order its method requires. Each algorithm has one warm-up and three timed runs
per input. Outside the timer, membership is checked against bitset DP
(analytically for the repeated generator), every parent edge is verified,
and five target witnesses are reconstructed.

We use GCC 13.3.0 with \texttt{-std=c++17 -O3 -DNDEBUG} on one pinned logical CPU
of an AMD Ryzen 7 5800H under Linux/WSL2. Time includes preprocessing,
allocation, and construction of the reachable set and parents. Input
generation and correctness checks are excluded. The plots use linear axes,
medians of three runs, and minimum-to-maximum bars for the same input.

\paragraph{Observations.}
Across the 30 plotted inputs, median runtimes give speedups of about
$4.5$--$22\times$ over the hashed shift tree and $18$--$48\times$ over the
deterministic shift tree (\cref{fig:runtime}). Bitset DP is faster on random
support and powers of three. Random support reaches full coverage after
a few updates. Inputs made of powers of three contain only $\Oh(\log m)$ items.
Both input families require few passes over the bitset. On repeated generators,
bitset DP performs $m-1$ full-bitset scans. This takes $\Theta(m^2/w)$ time.
\Cref{fig:linear-runtime} shows that our algorithm's
runtime grows roughly linearly with $m$ over the plotted range.

\clearpage
\begingroup
\setlength{\intextsep}{6pt}
\begin{figure}[H]
\centering
\setlength{\abovecaptionskip}{4pt}
\includegraphics[width=0.87\textwidth]{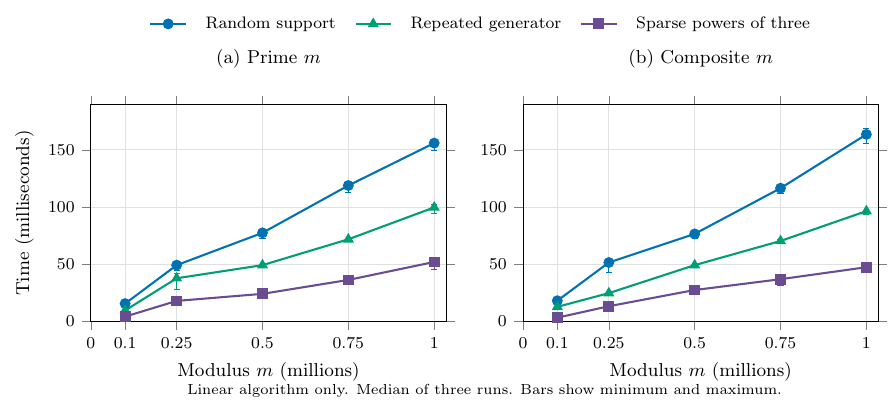}
\caption{Our linear algorithm in milliseconds. Prime and composite panels share a scale.}
\label{fig:linear-runtime}
\end{figure}

\begin{figure}[H]
\centering
\setlength{\abovecaptionskip}{4pt}
\includegraphics[width=\textwidth]{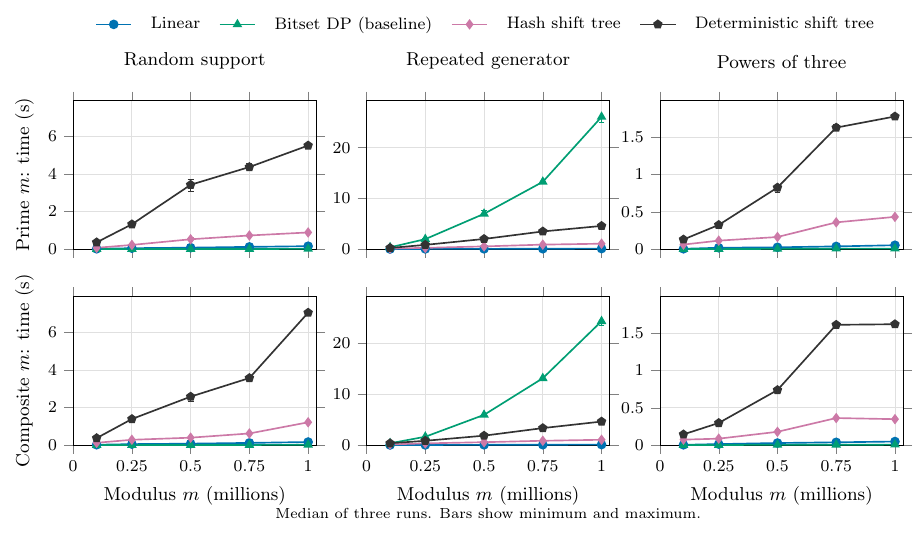}
\caption{All algorithms in seconds. Columns show the three input families. Each column has its own vertical scale. Rows show prime and composite moduli.}
\label{fig:runtime}
\end{figure}
\endgroup

\FloatBarrier

\section*{Acknowledgments}
We thank Krzysztof Pot\k{e}pa for reviewing the paper and suggesting that
cycle states be stored in packed bitsets, as in \cref{prop:update}, which
simplified the algorithm significantly.

\bibliographystyle{alpha}
\bibliography{references}

\end{document}